\documentclass[10pt,technote,letterpaper,twoside,twocolumn]{IEEEtran}

\usepackage{cite}
\usepackage{amsmath,amssymb,amsfonts}
\usepackage{graphicx}
\usepackage{multirow}
\usepackage{subcaption}
\usepackage{caption}
\usepackage{float}
\usepackage{textcomp}
\usepackage{xcolor}
\usepackage{tikz}
\usepackage{pgfplots}
\usepackage{stfloats}
\usepackage{mathrsfs}
\usepackage{balance}
\usepackage{amsthm} 
\usepackage[ruled,vlined]{algorithm2e}

\newtheorem{Proposition}{Proposition}
\usepackage{comment}
 \usepackage{booktabs}
\usepackage{algpseudocode}
\usepackage{acronym}
\usepackage[hyperfirst=true]{glossaries}
\usepackage[hidelinks]{hyperref} 
\usepackage{makecell}
\makeglossaries
\newacronym{OFDM}{OFDM}{orthogonal frequency division multiplexing}
\newacronym{OTFS}{OTFS}{orthogonal time frequency space}
\newacronym{AFDM}{AFDM}{affine frequency division multiplexing}
\newacronym{SIMO}{SIMO}{single-input multiple-output}
\newacronym{DL}{DL}{deep learning}
\newacronym{EVM}{EVM}{error vector magnitude}
\newacronym{BER}{BER}{bit error rate}
\newacronym{HSR}{HSR}{high-speed railways}
\newacronym{V2V}{V2V}{vehicle-to-vehicle}
\newacronym{V2I}{V2I}{vehicle-to-infrastructure}
\newacronym{V2X}{V2X}{vehicle-to-everything}
\newacronym{UAV}{UAV}{unmanned aerial vehicle}
\newacronym{ICI}{ICI}{intercarrier interference}
\newacronym{ISI}{ISI}{intersymbol interference}
\newacronym{IPI}{IPI}{interpath interference}
\newacronym{CFR}{CFR}{channel frequency response}
\newacronym{SISO}{SISO}{single-input single-output}
\newacronym{DoA}{DoA}{direction-of-arrival}
\newacronym{TX}{TX}{transmitter}
\newacronym{RX}{RX}{receiver}
\newacronym{CP}{CP}{cyclic prefix}
\newacronym{RMSE}{RMSE}{root mean squared error}
\newacronym{CRLB}{CRLB}{Cramér-Rao Lower Bound}
\newacronym{AWGN}{AWGN}{additive white gaussian noise}
\newacronym{SNR}{SNR}{signal-to-noise ratio}
\newacronym{MSE}{MSE}{mean squared error}
\newacronym{MLP}{MLP}{multilayer perceptron}
\newacronym{FIM}{FIM}{Fisher information matrix}
\newacronym{ULA}{ULA}{uniform linear array}
\newacronym{DFT}{DFT}{discrete Fourier transform}
\newacronym{QAM}{QAM}{quadrature amplitude modulation}
\newacronym{MF}{MF}{matched filter}
\newacronym{MRC}{MRC}{maximum ratio combining}
\newacronym{CFAR}{CFAR}{constant false alarm rate}
\newacronym{LS}{LS}{least squares}
\newacronym{FNN}{FNN}{feedforward neural network}
\newacronym{MP}{MP}{message passing}
\newacronym{RCP}{RCP}{reduced cyclic prefix}
\newacronym{DD}{DD}{delay-Doppler}
\newacronym{IDZT}{IDZT}{inverse discrete Zak transform}
\newacronym{DZT}{DZT}{discrete Zak transform}
\newacronym{PDR}{PDR}{pilot-to-data ratio}
\newacronym{PAPR}{PAPR}{peak-to-average power ratio}
\newacronym{MIMO}{MIMO}{multiple input multiple output}
\usepackage{orcidlink}
\usetikzlibrary{arrows.meta, positioning, fit, backgrounds, calc}

\allowdisplaybreaks

\begin{document}

\title{Superimposed Cross-Pilots: Addressing Fractional Shifts in DoA-Aided OTFS}

\author{
Mauro Marchese$^{\orcidlink{0009-0008-0265-5840}}$, Pietro Savazzi$^{\orcidlink{0000-0003-0692-8566}}$
\thanks{Mauro Marchese is with the Department of Electrical, Computer and Biomedical Engineering, University of Pavia, Pavia, 27100 Italy  (e-mail: mauro.marchese01@universitadipavia.it).}
\thanks{Pietro Savazzi is with the Department of Electrical, Computer and Biomedical Engineering, University of Pavia, Pavia, 27100 Italy (e-mail: pietro.savazzi@unipv.it), and with the Consorzio Nazionale Interuniversitario per le Telecomunicazioni - CNIT.}
}

\maketitle

\begin{abstract}
In this work, a novel superimposed pilot scheme, named superimposed cross-pilots, is proposed for fractional parameter estimation in multi-antenna \gls{OTFS} receivers. Assuming a large \gls{ULA} size at the receiver, the multipath components are separated in the angular domain through a \gls{MF}. It is then shown that the proposed superimposed pilot scheme enables the computation of integrated delay and Doppler profiles by averaging the received delay-Doppler matrix across the Doppler and delay axes, respectively. This procedure helps reduce data-to-pilot interference via data averaging, eliminating the need for iterative cancellation schemes. Based on this, a fractional parameter estimation algorithm, which exploits \glspl{MF}, is derived. Simulation results show that the proposed approach outperforms existing \gls{OTFS} superimposed pilot schemes, achieving a lower \gls{BER} while exhibiting a trade-off between \gls{PAPR} and communication performance.
\end{abstract}


\glsresetall

\section{Introduction}
\IEEEPARstart{A}{ngle}-domain processing allows for the separation of multipath components when the receiver is equipped with a large antenna array \cite{Ge2019,Guo2017,marchese20266gofdmcommunicationshigh,Cheng2020}, thereby transforming a doubly-dispersive channel into multiple parallel flat-fading channels. In \cite{Ge2019,Guo2017,marchese20266gofdmcommunicationshigh}, multi-antenna \gls{OFDM} receivers are proposed for reliable communications under receiver mobility \cite{Ge2019,Guo2017} and the mobility of both transceivers and scatterers \cite{marchese20266gofdmcommunicationshigh}. Similarly, as the \gls{OTFS} waveform proposed in \cite{Hadani2017} has gained significant attention for the development of next-generation wireless systems \cite{HadaniMonk2018,Mohammed2022,Zhou2024,Hong2026,Nie2026}, a multi-antenna receiver is designed in \cite{Cheng2020} to separate multipaths in the angular domain. This approach increases sparsity in the delay-Doppler domain and reduces the data detection complexity of the \gls{MP} algorithm. However, \cite{Cheng2020} assumes ideal channel estimation.
The channel estimation problem in \gls{OTFS} has been addressed in several works \cite{Raviteja2019,Khan2023,Yogesh2024,Muppaneni2023,marchese2024,marchese2025disjoint,Yuan2021,Mishra2022,Jesbin2023,kanazawa2025}. Embedded-pilot and full-pilot schemes have been investigated in \cite{Raviteja2019,Khan2023,Yogesh2024} and \cite{Muppaneni2023,marchese2024,marchese2025disjoint}, respectively. The algorithm in \cite{Raviteja2019} assumes integer delays and Doppler shifts, thereby exhibiting performance degradation in real-world scenarios with fractional delays and Doppler shifts \cite{Muppaneni2023,marchese2024,Yogesh2024,marchese2025disjoint}. The problem of fractional parameter estimation has been addressed in \cite{Khan2023} and further investigated in \cite{Muppaneni2023,Yogesh2024,marchese2024,marchese2025disjoint}.
Specifically, in \cite{Khan2023,Yogesh2024,marchese2025disjoint}, disjoint delay-Doppler estimation is performed, reducing complexity compared to joint delay-Doppler estimation \cite{Muppaneni2023,marchese2024}. Superimposed pilot schemes have been introduced in the literature \cite{Yuan2021,Mishra2022,Jesbin2023,kanazawa2025} to increase spectral efficiency compared to embedded-pilot and full-pilot approaches. Single-pilot and multiple-pilot schemes have been proposed in \cite{Yuan2021} and \cite{kanazawa2025}, respectively. Moreover, in \cite{kanazawa2025}, it is shown that the multiple superimposed pilot scheme outperforms embedded-pilot and single superimposed pilot schemes by achieving higher throughput. However, in superimposed pilot approaches, it is typically assumed that the channel consists of integer delays and Doppler shifts \cite{Yuan2021,Mishra2022,Jesbin2023,kanazawa2025}.

In light of the above discussion, this work proposes a novel superimposed pilot scheme, termed \textit{superimposed cross-pilots}, for \gls{DoA}-aided \gls{OTFS}-based receivers. It is demonstrated that the proposed scheme enables the computation of integrated delay and Doppler profiles through averaging. This procedure effectively reduces data-to-pilot interference, eliminating the need for iterative schemes required by state-of-the-art superimposed pilot approaches \cite{kanazawa2025}. Finally, based on the proposed scheme, a low-complexity disjoint fractional delay-Doppler estimation algorithm is designed.

\section{System Model}
The scenario includes a single-antenna \gls{TX} and a multi-antenna \gls{RX} equipped with a \gls{ULA} with $N_r$ receiving antennas and working at carrier frequency $f_c$.
The \gls{TX} sends to the \gls{RX} an \gls{OTFS} frame including $M$ delay bins and $N$ Doppler bins. Therefore, the \gls{OTFS} transmit signal is made of $M$ subcarriers with spacing $\Delta f=1/T$, where $T$ is the symbol duration, and $N$ time slots. Thus, the signal bandwidth is $B=M\Delta f$. The \gls{TX} adopts \gls{CP}-\gls{OTFS}, meaning that each block of the transmit \gls{OTFS} frame is preceded by a \gls{CP} with duration $T_{PC}>\sigma_\tau$ to prevent \gls{ISI}, where $\sigma_\tau$ is the channel delay spread. Hence, the overall symbol duration is $T'=T_{CP}+T$ and the frame duration is $T_f=NT'$. The system model is depicted in Figure \ref{fig:system_model}, illustrating the \gls{TX} and \gls{DoA}-aided \gls{RX} processing chains.

\subsection{Single-Antenna OTFS Transmitter}
The \gls{OTFS} modulator arranges $MN$ symbols in the \gls{DD} domain over the two-dimensional grid $I_{DD}=\big\{m\Delta\tau, n\Delta\nu \ | \ 0 \leq m \leq M-1, \ 0 \leq n \leq N-1\big\}$, where $\Delta\tau = 1/B$ and $\Delta\nu = 1/T_f$ represent the delay and Doppler resolutions, respectively. Hence, the \gls{DD} matrix $\mathbf{X} \in \mathbb{C}^{M \times N}$ is obtained. 
The \gls{TX} uses $N_p$ superimposed pilots for channel estimation. Therefore,
\begin{equation} \label{eq:SuperimposedPilotModel}
\mathbf{X}=\sqrt{E_s}\mathbf{X}_d+\sqrt{E_p}\mathbf{X}_p \in \mathbb{C}^{M \times N},
\end{equation}
where $\mathbf{X}_d$ contains \gls{QAM} data with $\mathbb{E}\big[\big|[\mathbf{X}_d]_{m,n}\big|^2\big]=1$ and $\mathbf{X}_p$ contains pilot symbols such that $[\mathbf{X}_p]_{m,n}=1$ if $(m,n)=\big\{(m_p^{(1)},n_p^{(1)}),(m_p^{(2)},n_p^{(2)}),\dots,(m_p^{(N_p)},n_p^{(N_p)})\big\}$, and $[\mathbf{X}_p]_{m,n}=0$ otherwise. Moreover, $E_s$ is the average energy per symbol allocated for data and $E_p$ is the pilot energy. The total energy is split between pilot and data by fixing the \gls{PDR}, which is defined as $\text{PDR}=E_p/E_s$. The \gls{TX} is subject to an average power constraint so that to total energy per frame $E_f=MNE_s+N_pE_p$ is fixed.

The transmit vector is obtained via \gls{IDZT} as
\begin{equation} \label{eq:TransmitSamplesVector}
\mathbf{s}=\big(\mathbf{F}_{N}^{\mathsf{H}}\otimes\mathbf{I}_{M}\big)\mathbf{x}\in \mathbb{C}^{MN},
\end{equation}
where $[\mathbf{F}_N]_{p,q}=\frac{1}{\sqrt{N}}e^{-j2\pi\frac{pq}{N}}$ is the $N$-point unitary \gls{DFT} matrix, $\mathbf{I}_{M}$ is the identity matrix of order $M$ and $\textbf{x}=\text{vec}(\mathbf{X})$.

\usetikzlibrary{shapes.symbols}
\begin{figure*}[!t]
\centering
 
\tikzset{
  txblk/.style={
    draw, fill=blue!12, rounded corners=3pt,
    minimum width=2.5cm, minimum height=1.45cm,
    align=center, font=\small, line width=0.45pt
  },
  mfblk/.style={
    draw, fill=teal!22, rounded corners=3pt,
    minimum width=2cm, minimum height=0.9cm,
    align=center, font=\small, inner sep=3pt, line width=0.45pt
  },
  prblk/.style={
    draw, fill=violet!15, rounded corners=3pt,
    minimum width=2cm,  minimum height=0.9cm,
    align=center, font=\small, inner sep=3pt, line width=0.45pt
  },
  estblk/.style={
    draw, fill=orange!20, rounded corners=3pt,
    minimum width=2.5cm,  minimum height=0.9cm,
    align=center, font=\small, inner sep=3pt, line width=0.45pt
  },
  rarr/.style={
    -{Latex[length=2pt, width=3.5pt]},
    line width=0.55pt
  },
}
 
\begin{tikzpicture}

\node[font=\small\bfseries] (txtitle) at (0, 1.25) {OTFS TX};

\node[draw, fill=gray!18, rounded corners=2pt, line width=0.4pt,
      font=\scriptsize, align=center,
      minimum width=2.0cm, minimum height=0.82cm] (xd) at (-1.1, 0.42)
  {\textbf{Data matrix}\\[2pt]%
   $\mathbf{X}_d \in \mathbb{C}^{M\times N}$};

\node[draw, fill=orange!28, rounded corners=2pt, line width=0.4pt,
      font=\scriptsize, align=center,
      minimum width=2.0cm, minimum height=0.82cm] (xp) at (1.1, 0.42)
  {\textbf{Cross-pilots}\\[2pt]%
   $\mathbf{X}_p$};

\node[draw, fill=blue!7, rounded corners=2pt, line width=0.35pt,
      font=\small, align=center,
      minimum width=4.3cm, minimum height=0.65cm] (xsum) at (0, -0.72)
  {$\mathbf{X}=\sqrt{E_s}\,\mathbf{X}_d+\sqrt{E_p}\,\mathbf{X}_p$};

\node[draw, fill=green!12, rounded corners=2pt, line width=0.4pt,
      font=\small, align=center,
      minimum width=4.3cm, minimum height=0.65cm] (idzt) at (0, -1.78)
  {\textbf{IDZT}\;\;$\mathbf{s}=(\mathbf{F}_N^{\mathsf{H}}\otimes\mathbf{I}_M)\,\mathrm{vec}(\mathbf{X})$};

\draw[rarr] (xp.south) -- (xsum.north);
\draw[rarr] (xd.south) -- (xsum.north);
\draw[rarr] (xsum.south) -- (idzt.north);

\begin{pgfonlayer}{background}
  \node[draw, fill=blue!9, rounded corners=4pt, line width=0.45pt,
        fit=(txtitle)(xd)(xp)(idzt), inner sep=7pt] (tx) {};
\end{pgfonlayer}

\node[minimum width=0.42cm, minimum height=5.0cm, inner sep=0]
  (ula) at ($(tx.east)+(4.3,0)$) {};
 
\draw[draw=teal!65, fill=teal!14, rounded corners=2pt, line width=0.45pt]
  (ula.south west) rectangle (ula.north east);

\node[font=\small\bfseries, rotate=90, text=teal!80!black]
  at (ula.center) {ULA};
  

\node[cloud, cloud puffs=12, cloud puff arc=100,
      draw=black, fill=blue!5, line width=0.5pt,
      font=\scriptsize, align=center,
      inner sep=1pt,
      scale=0.7, transform shape]
  (channel) at ($(tx.east)!0.50!(ula.west)$)
  {\normalsize Channel\\[3pt]\normalsize $\{\theta_p,\,l_p,\,k_p,\,\alpha_p\}_{p=1}^{P}$};
  
\draw[-{Latex[length=4pt, width=5pt]}, line width=0.8pt]
  (tx.east) -- (channel);
\draw[-{Latex[length=4pt, width=5pt]}, line width=0.8pt]
  (channel) -- (ula.west);

\node[mfblk]  (mf1) at ($(ula.east)+(1.55, 1.5)$)
  {\textbf{MF($\theta_1$)\,+\,DZT}};
\node[prblk,  right=0.35cm of mf1] (pr1)
  {\textbf{DD profiles}\\\scriptsize $\mathbf{u}_1,\mathbf{v}_1$\;};
\node[estblk, right=0.35cm of pr1] (es1)
  {\textbf{DD estimation}\\\scriptsize $\hat{l}_1,\,\hat{k}_1,\,\hat{\alpha}_1$};
\draw[rarr] ($(ula.east)+(0, 1.5)$) -- (mf1.west);
\draw[rarr] (mf1.east) -- (pr1.west);
\draw[rarr] (pr1.east) -- (es1.west);
 
\node[mfblk]  (mf2) at ($(ula.east)+(1.55, 0.5)$)
  {\textbf{MF($\theta_2$)\,+\,DZT}};
\node[prblk,  right=0.35cm of mf2] (pr2)
  {\textbf{DD profiles}\\\scriptsize $\mathbf{u}_2,\mathbf{v}_2$\;};
\node[estblk, right=0.35cm of pr2] (es2)
  {\textbf{DD estimation}\\\scriptsize $\hat{l}_2,\,\hat{k}_2,\,\hat{\alpha}_2$};
\draw[rarr] ($(ula.east)+(0, 0.5)$) -- (mf2.west);
\draw[rarr] (mf2.east) -- (pr2.west);
\draw[rarr] (pr2.east) -- (es2.west);
 
\node at (mf2 |- {(0,-0.5)}) {\Large$\vdots$};
\node at (pr2 |- {(0,-0.5)}) {\Large$\vdots$};
\node at (es2 |- {(0,-0.5)}) {\Large$\vdots$};

\node[mfblk]  (mf4) at ($(ula.east)+(1.55,-1.5)$)
  {\textbf{MF($\theta_P$)\,+\,DZT}};
\node[prblk,  right=0.35cm of mf4] (pr4)
  {\textbf{DD profiles}\\\scriptsize $\mathbf{u}_P,\mathbf{v}_P$\;};
\node[estblk, right=0.35cm of pr4] (es4)
  {\textbf{DD estimation}\\\scriptsize $\hat{l}_P,\,\hat{k}_P,\,\hat{\alpha}_P$};
\draw[rarr] ($(ula.east)+(0,-1.5)$) -- (mf4.west);
\draw[rarr] (mf4.east) -- (pr4.west);
\draw[rarr] (pr4.east) -- (es4.west);
 
\begin{pgfonlayer}{background}
  \node[draw, dashed, rounded corners=5pt,
        inner sep=8pt, line width=0.5pt,
        fit=(mf1)(mf4)(es1)(es4),
        label={[font=\small\bfseries]above:%
          DoA-Aided OTFS Receiver}
       ] (bbox) {};
\end{pgfonlayer}
 
\end{tikzpicture}
 
\caption{System model of the proposed \gls{DoA}-aided \gls{OTFS} receiver with superimposed cross-pilots. The single-antenna transmitter sends an \gls{OTFS} frame $\mathbf{X}$ superimposing pilot $\mathbf{X}_p$ and data $\mathbf{X}_d$. After propagation through $P$ multipath components
(each characterised by \gls{DoA} $\theta_p$, normalized delay $l_p$,
normalized Doppler $k_p$, and complex gain $\alpha_p$), the \gls{ULA} separates the paths in the angular domain. Each branch applies a \gls{MF} followed by a \gls{DZT}, computes integrated delay/Doppler profiles $\mathbf{u}_p$, $\mathbf{v}_p$ via averaging, and outputs disjoint fractional estimates $\hat{l}_p$, $\hat{k}_p$, $\hat{\alpha}_p$.}
\label{fig:system_model}
\end{figure*}

\subsection{Observation Model at the Multi-Antenna Receiver}
The transmit signal passes through a doubly-dispersive wireless channel made of $P$ propagation paths with delay $\tau_p$, Doppler shift $\nu_p$, channel gain $\alpha_p$ and \gls{DoA} $\theta_p$. According to \cite{marchese2025disjoint}, the received time-spatial observations $\mathbf{R}\in\mathbb{C}^{MN\times N_r}$ are obtained as
\begin{equation} \label{eq:timespatialobs}
\mathbf{R}=\sum_{p=1}^P\alpha_p\boldsymbol{\Delta}(k_p)\big(\mathbf{I}_N \otimes \mathbf{C}_M(l_p)\big)\mathbf{s}\mathbf{a}_{rx}^\top(\theta_p)+\mathbf{N},
\end{equation}
where $l_p=\tau_p/\Delta\tau$ and $k_p=\nu_p/\Delta\nu$ are the normalized delay and Doppler shifts, respectively. Moreover, $\mathbf{N}\in\mathbb{C}^{MN\times N_r}$ is the \gls{AWGN} matrix and $\text{vec}(\mathbf{N})\sim\mathcal{CN}(0,\sigma^2\mathbf{I}_{MNN_r})$. Noise variance is given as $\sigma^2=N_0$ where $N_0$ is the noise power spectral density. Further, $\mathbf{a}_{rx}(\theta)$ is the steering vector of the \gls{ULA} at the \gls{RX}. It is computed as $\mathbf{a}_{rx}(\theta)=\big[1 \ \ e^{j\frac{2\pi}{\lambda}d\sin(\theta)} \ \ \dots \ \ e^{j\frac{2\pi}{\lambda}d(N_r-1)\sin(\theta)}\big]^\top$ assuming the antenna spacing at the \gls{RX} is $d=\lambda/2$, where $\lambda=c/f_c$ is the wavelength and $c$ denotes the speed of light. Finally, the Doppler matrix $\boldsymbol{\Delta}(k)\in\mathbb{C}^{MN\times MN}$ is given as
\begin{equation}
    \boldsymbol{\Delta}(k_p)=\mathbf{D}_N^{k_p}\otimes\tilde{\mathbf{D}}_M^{k_p},
\end{equation}
where $\mathbf{D}_N\in\mathbb{C}^{N\times N}$ and $\tilde{\mathbf{D}}_M\in\mathbb{C}^{M\times M}$ are diagonal matrices with $[\mathbf{D}_N]_{n,n}=e^{j\frac{2\pi}{N}n}$ and $[\tilde{\mathbf{D}}_M]_{m,m}=e^{j\frac{2\pi}{MN}\frac{mT}{T'}}$, respectively. Moreover, $\mathbf{D}_N$ and $\tilde{\mathbf{D}}_M$ capture intersymbol and intrasymol (i.e., \gls{ICI}) Doppler-induced phase rotations, respectively.
The delay matrix $\mathbf{C}_M(l_p)\in\mathbb{C}^{M\times M}$ is a circulant matrix defined as
\begin{equation}
    \mathbf{C}_M(l_p)=\mathbf{F}_M^{\mathsf{H}} (\mathbf{D}_M^*)^{l_p} \mathbf{F}_M,
\end{equation}
where $\mathbf{D}_M\in\mathbb{C}^{M\times M}$ is a diagonal matrix with $[\mathbf{D}_M]_{m,m}=e^{j\frac{2\pi}{M}m}$.
The \gls{SNR} is obtained as 
\begin{equation}
\text{SNR}=\frac{E_f}{MN\sigma^2}=\frac{MNE_s+N_pE_p}{MN\sigma^2}.
\end{equation}

 \subsection{Angle-Domain Beamforming}
 Hereafter, the following assumptions are made: (i) the \glspl{DoA} are known due to previous estimation\footnote{The estimation of \glspl{DoA} is beyond the scope of this work, which focuses on fractional delay and Doppler estimation for superimposed pilot-based angle-domain \gls{OTFS} receivers. Since the proposed approach and the baseline method are evaluated under identical operating conditions, any impairment from imperfect \gls{DoA} knowledge affects both equally, ensuring a fair comparison.}, (ii) the receiver is equipped with a large \gls{ULA} (the number of antennas is sufficienlty high), so that the \gls{IPI} between multipaths with different angles is negligible \cite{Ge2019,marchese20266gofdmcommunicationshigh}.
The first operation made by angle-domain based receivers is multipath separaion by means of a \gls{MF} \cite{Ge2019,marchese20266gofdmcommunicationshigh}. Using the fact that $\mathbf{a}_{\text{rx}}^\top(\theta_1)\mathbf{a}_{\text{rx}}^*(\theta_2)\approx 0$ if $\theta_1\neq\theta_2$ and $\mathbf{a}_{\text{rx}}^\top(\theta_1)\mathbf{a}_{\text{rx}}^*(\theta_2)=N_r$ if $\theta_1=\theta_2$ for sufficiently high $N_r$ \cite{Ge2019,marchese20266gofdmcommunicationshigh}, multipaths are separated in the angular domain performing
\begin{equation}\label{eq:rp_beamforming}
\mathbf{r}_p=\frac{\mathbf{R}\mathbf{a}^*_{rx}(\theta_p)}{N_r}\approx\alpha_p\boldsymbol{\Delta}(k_p)\big(\mathbf{I}_N \otimes \mathbf{C}_M(l_p)\big)\mathbf{s}+\mathbf{n_p},
\end{equation}
where $\mathbf{n}_p=\mathbf{N}\mathbf{a}^*_{rx}(\theta_p)/N_r$. Hence, $\mathbf{n}_p\sim\mathcal{CN}(\mathbf{0}_{MN},\sigma^2\mathbf{I}_{MN}/N_r)$.
After that, the signal in \eqref{eq:rp_beamforming} is converted to delay-Doppler domain via \gls{DZT} as
\begin{equation}\begin{split}\label{eq:delayDopplerObserv}
\mathbf{y}_p=&\big(\mathbf{F}_{N}\otimes\mathbf{I}_{M}\big)\mathbf{r}_p
\\
\approx&\alpha_p\big(\mathbf{F}_{N}\otimes\mathbf{I}_{M}\big)\boldsymbol{\Delta}(k_p)\big(\mathbf{F}_{N}^{\mathsf{H}}\otimes\mathbf{C}_M(l_p)\big)\mathbf{x}+\mathbf{z}_p
\\
=&\alpha_p\big(\mathbf{F}_N\mathbf{D}_N^{k_p}\mathbf{F}_{N}^{\mathsf{H}}\otimes\tilde{\mathbf{D}}_M^{k_p}\mathbf{C}_M(l_p)\big)\mathbf{x}+\mathbf{z}_p,
\end{split}\end{equation}
where $\mathbf{z}_p=\big(\mathbf{F}_{N}\otimes\mathbf{I}_{M}\big)\mathbf{n}_p$ is the delay-Doppler domain \gls{AWGN} vector.

\begin{figure}[t]
 \centering
 \begin{subfigure}[t]{0.24\textwidth} 
 \centering
 \resizebox{0.99\columnwidth}{!}{
 \input{Results/multiplepilotsscheme.tikz}}
 \caption{}\label{fig:multipilotscheme}
 \end{subfigure}%
 \hfill 
 \begin{subfigure}[t]{0.24\textwidth} 
 \centering
 \resizebox{0.99\columnwidth}{!}{
 \input{Results/crosspilotscheme.tikz}}
 \caption{}\label{fig:crosspilotscheme}
 \end{subfigure}%
 \caption{Pilot matrix $\mathbf{X}_p \in \mathbb{C}^{M \times N}$ in the delay-Doppler (DD)
domain ($M=64$ delay bins, $N=16$ Doppler bins) for the superimposed pilot configurations compared in this work:
(a) multiple superimposed pilot scheme of \cite{kanazawa2025}, where
pilots are placed at $N_p$ isolated DD grid locations
$\{(m_p^{(i)},n_p^{(i)})\}_{i=1}^{N_p}$;
(b) proposed superimposed cross-pilot scheme, where pilots are allocated on a full delay row (index $m_p$, across all $N$ Doppler bins) and a full Doppler column (index $n_p$, across all $M$ delay bins),
yielding $N_p = M+N-1$ pilots in a cross-shaped pattern.} 
 \label{fig:PilotMatrices}
\end{figure}

\section{Proposed Superimposed Pilot Scheme for DoA-Aided Receivers}
This section presents the proposed superimposed pilot scheme and the proposed fractional delay-Doppler estimation algorithm.

\subsection{Proposed Superimposed Cross-Pilots}
Figure \ref{fig:multipilotscheme} shows the multiple superimposed pilot scheme adopted in \cite{kanazawa2025}. A limitation of the multiple superimposed pilot scheme is that, in the presence of fractional channel parameters, the pilots interfere with each other due to the spreading effect \cite{marchese2025robust6gofdmhighmobility,marchese2025disjoint,Muppaneni2023}. In \gls{DoA}-aided receivers, the following superimposed pilot scheme can be adopted, where the delay-Doppler pilots are allocated as
\begin{equation}\label{eq:crosspilots}
\mathbf{x}_p=(\mathbf{1}_N\otimes\mathbf{e}_{m_p})+\mathbf{e}_{n_p}\otimes(\mathbf{1}_M-\mathbf{e}_{m_p}),
\end{equation}
where $m_p$ and $n_p$ represent the delay and Doppler indices over which the pilots are superimposed, respectively. Consequently, the number of superimposed pilots is $N_p=M+N-1$. Figure \ref{fig:crosspilotscheme} illustrates an example of the proposed pilot scheme compared to the multiple superimposed pilot scheme in \cite{kanazawa2025}.

\subsection{Disjoint Delay-Doppler Estimation Via Averaging}
Given the proposed superimposed cross-pilot scheme, it is possible to estimate delay and Doppler shifts separately as follows. The delay profile is obtained by averaging the columns of the delay-Doppler matrix $\mathbf{Y}_p=\text{vec}^{-1}(\mathbf{y}_p)$ as
\begin{equation}\label{eq:delayprofile}
\mathbf{u}_p=\frac{\mathbf{Y}_p\mathbf{1}_N}{N}=\frac{(\mathbf{1}_N^\top\otimes\mathbf{I}_M)\mathbf{y}_p}{N}\in\mathbb{C}^M.
\end{equation}
Similarly, the Doppler profile is obtained by summing up the rows of $\mathbf{Y}_p$ as
\begin{equation}\label{eq:Dopplerprofile}
\mathbf{v}_p=\frac{\mathbf{Y}_p^\top\mathbf{1}_M}{M}=\frac{(\mathbf{I}_N\otimes\mathbf{1}^\top_M)\mathbf{y}_p}{M}\in\mathbb{C}^N.
\end{equation}
\begin{Proposition}\label{Proposition:delayprofile}
The integrated delay profile in \eqref{eq:delayprofile} is given by\footnote{The subscript $p$, indicating the path brench, is omitted to lighten the notation in both the statement of the Proposition and the proof.}
\begin{equation}\label{eq:delayProfileProposition}
\mathbf{u}=\alpha\Big({\frac{\sqrt{E_p}}{N}}\mathbf{g}_u(l,k)+{\frac{\sqrt{E_s}}{N}}\tilde{\mathbf{D}}_M^{k}\mathbf{C}_M(l)\sum_{n=0}^{N-1}[\mathbf{X}_d]_{:,n}\Big)+\tilde{\mathbf{z}}_u,
\end{equation}
where $\tilde{\mathbf{z}}_u\sim\mathcal{CN}\big(\mathbf{0}_M,\frac{\sigma^2}{NN_r}\mathbf{I}_M\big)$ and
\begin{equation}
\mathbf{g}_u(l,k)=(N-1)\tilde{\mathbf{D}}_M^k\mathbf{F}_M^{\mathsf{H}} \Big( [\mathbf{F}_M]_{:,m_p} \odot \mathbf{d}_M^{-l} \Big)+\tilde{\mathbf{d}}_M^k
\end{equation}
with $\mathbf{d}_M=\textnormal{diag}\big(\mathbf{D}_M\big)$ and $\tilde{\mathbf{d}}_M=\textnormal{diag}\big(\tilde{\mathbf{D}}_M\big)$.
\end{Proposition}

\begin{proof}
The proof relies on the Kronecker product properties and on the \gls{DFT} property $\mathbf{F}_N\mathbf{e}_0=\mathbf{1}_N/\sqrt{N}$.
Combining \eqref{eq:SuperimposedPilotModel}, \eqref{eq:delayDopplerObserv}, \eqref{eq:crosspilots} and plugging \eqref{eq:delayDopplerObserv} into \eqref{eq:delayprofile}, the term related to the pilot in the delay profile becomes proportional to the following terms:
\begin{equation}\begin{split}
&(\mathbf{1}_N^\top\otimes\mathbf{I}_M)\big(\mathbf{F}_N\mathbf{D}_N^{k}\mathbf{F}_{N}^{\mathsf{H}}\otimes\tilde{\mathbf{D}}_M^{k}\mathbf{C}_M(l)\big)(\mathbf{1}_N\otimes\mathbf{e}_{m_p})=
\\
&\mathbf{1}_N^\top\mathbf{F}_N\mathbf{D}_N^{k}\mathbf{F}_{N}^{\mathsf{H}}\mathbf{1}_N\otimes\tilde{\mathbf{D}}_M^{k}\mathbf{C}_M(l)\mathbf{e}_{m_p}=
\\
&N\mathbf{e}_0^\top\mathbf{D}_N^{k}\mathbf{e}_0\tilde{\mathbf{D}}_M^{k}\mathbf{C}_M(l)\mathbf{e}_{m_p}=N\tilde{\mathbf{D}}_M^{k}\mathbf{C}_M(l)\mathbf{e}_{m_p}=
\\
&N\tilde{\mathbf{D}}_M^k\mathbf{F}_M^{\mathsf{H}} \big( [\mathbf{F}_M]_{:,m_p} \odot \mathbf{d}_M^{-l} \big),
\end{split}
\end{equation}

\begin{equation}\begin{split}
&(\mathbf{1}_N^\top\otimes\mathbf{I}_M)\big(\mathbf{F}_N\mathbf{D}_N^{k}\mathbf{F}_{N}^{\mathsf{H}}\otimes\tilde{\mathbf{D}}_M^{k}\mathbf{C}_M(l)\big)(\mathbf{e}_{n_p}\otimes\mathbf{1}_M)=
\\
&\mathbf{1}_N^\top\mathbf{F}_N\mathbf{D}_N^{k}[\mathbf{F}_N^\mathsf{H}]_{:,n_p}\otimes\tilde{\mathbf{D}}_M^{k}\mathbf{F}_M^{\mathsf{H}} (\mathbf{D}_M^*)^{l}\mathbf{F}_M\mathbf{1}_M=
\\
&\sqrt{MN}\mathbf{e}_0^\top[\mathbf{F}_N^\mathsf{H}]_{:,n_p}\tilde{\mathbf{D}}_M^{k}\mathbf{F}_M^{\mathsf{H}} (\mathbf{D}_M^*)^{l}\mathbf{e}_0=\tilde{\mathbf{D}}_M^{k}\mathbf{1}_M,
\end{split}
\end{equation}

\begin{equation}\begin{split}
&(\mathbf{1}_N^\top\otimes\mathbf{I}_M)\big(\mathbf{F}_N\mathbf{D}_N^{k}\mathbf{F}_{N}^{\mathsf{H}}\otimes\tilde{\mathbf{D}}_M^{k}\mathbf{C}_M(l)\big)(\mathbf{e}_{n_p}\otimes\mathbf{e}_{m_p})=
\\
&\tilde{\mathbf{D}}_M^{k}\mathbf{C}_M(l)\mathbf{e}_{m_p}=\tilde{\mathbf{D}}_M^k\mathbf{F}_M^{\mathsf{H}} \big( [\mathbf{F}_M]_{:,m_p} \odot \mathbf{d}_M^{-l} \big).
\end{split}
\end{equation}
Moreover, the term related to data becomes proportional to
\begin{equation}\begin{split}
&(\mathbf{1}_N^\top\otimes\mathbf{I}_M)\big(\mathbf{F}_N\mathbf{D}_N^{k}\mathbf{F}_{N}^{\mathsf{H}}\otimes\tilde{\mathbf{D}}_M^{k}\mathbf{C}_M(l)\big)\mathbf{x}_d=
\\
&\big(\sqrt{N}\mathbf{e}_0^\top\mathbf{D}_N^{k}\mathbf{F}_N^\mathsf{H}\otimes\tilde{\mathbf{D}}_M^{k}\mathbf{C}_M(l)\big)\mathbf{x}_d=\big(\mathbf{1}_N^\top\otimes\tilde{\mathbf{D}}_M^{k}\mathbf{C}_M(l)\big)\mathbf{x}_d,
\end{split}
\end{equation}
where the last equality implies that the data term is proportional to the average of the columns of the data symbol matrix $\mathbf{X}_d$. This averaging reduces the power of the interfering data by a factor of $N$. In fact, as the frame size increases, the data-to-pilot interference tends to zero
\begin{equation}
\lim_{N\rightarrow\infty}\frac{1}{N}\sum_{n=0}^{N-1}[\mathbf{X}_d]_{:,n}=\mathbb{E}\Big[[\mathbf{X}_d]_{:,n}\Big]=\mathbf{0}_M.
\end{equation}
The same averaging procedure applies to the noise, thereby reducing the noise power by a factor of $N$.
\end{proof}

\begin{Proposition}\label{Proposition:dopplerprofile}
Assuming $k\ll N$ (negligible \gls{ICI}), the integrated Doppler profile in \eqref{eq:Dopplerprofile} can be approximated as 
\begin{equation}\label{eq:DopplerProfileProposition}
\mathbf{v}\approx\alpha\Big({\frac{\sqrt{E_p}}{M}}\mathbf{g}_v(k)+{\frac{\sqrt{E_s}}{M}}\mathbf{F}_N\mathbf{D}_N^k\mathbf{F}_N^\mathsf{H}\sum_{m=0}^{M-1}\big[\mathbf{X}_d^\top\big]_{:,m}\Big)+\tilde{\mathbf{z}}_v,
\end{equation}
where $\tilde{\mathbf{z}}_v\sim\mathcal{CN}(\mathbf{0}_N,\frac{\sigma^2}{MN_r}\mathbf{I}_N)$ and
\begin{equation}
\mathbf{g}_v(k)=(M-1)\mathbf{F}_N \Big( \big[\mathbf{F}_N^\mathsf{H}\big]_{:,n_p} \odot \mathbf{d}_N^{k} \Big)+\mathbf{1}_N
\end{equation}
with $\mathbf{d}_N=\textnormal{diag}\big(\mathbf{D}_N\big)$.
\end{Proposition}

\begin{proof}
The proof relies on the approximation $\tilde{\mathbf{D}}_M^{k}\approx\mathbf{I}_M$, which holds if $k\ll N$.
Combining \eqref{eq:SuperimposedPilotModel}, \eqref{eq:delayDopplerObserv} and \eqref{eq:crosspilots}, and plugging \eqref{eq:delayDopplerObserv} into \eqref{eq:Dopplerprofile}, the term related to the pilot in the Doppler profile becomes proportional to the following terms:
\begin{equation}\begin{split}
&(\mathbf{I}_N\otimes\mathbf{1}^\top_M)\big(\mathbf{F}_N\mathbf{D}_N^{k}\mathbf{F}_{N}^{\mathsf{H}}\otimes\tilde{\mathbf{D}}_M^{k}\mathbf{C}_M(l)\big)(\mathbf{1}_N\otimes\mathbf{e}_{m_p})\approx
\\
&\mathbf{F}_N\mathbf{D}_N^{k}\mathbf{F}_{N}^{\mathsf{H}}\mathbf{1}_N\otimes\sqrt{M}\mathbf{e}_0^\top(\mathbf{D}_M^*)^{l} \mathbf{F}_M\mathbf{e}_{m_p}=
\\
&\sqrt{N}\mathbf{F}_N\mathbf{D}_N^{k}\mathbf{e}_0=\mathbf{1}_N,
\end{split}
\end{equation}

\begin{equation}\begin{split}
&(\mathbf{I}_N\otimes\mathbf{1}^\top_M)\big(\mathbf{F}_N\mathbf{D}_N^{k}\mathbf{F}_{N}^{\mathsf{H}}\otimes\tilde{\mathbf{D}}_M^{k}\mathbf{C}_M(l)\big)(\mathbf{e}_{n_p}\otimes\mathbf{1}_M)\approx
\\
&\mathbf{F}_N\mathbf{D}_N^{k}\mathbf{F}_N^\mathsf{H}\mathbf{e}_{n_p}\otimes\mathbf{1}^\top_M\mathbf{F}_M^{\mathsf{H}} (\mathbf{D}_M^*)^{l}\mathbf{F}_M\mathbf{1}_M=
\\
&M\mathbf{e}_0^\top(\mathbf{D}_M^*)^{l}\mathbf{e}_0\mathbf{F}_N\mathbf{D}_N^{k}[\mathbf{F}_N^\mathsf{H}]_{:,n_p}=M\mathbf{F}_N \big( \big[\mathbf{F}_N^\mathsf{H}\big]_{:,n_p} \odot \mathbf{d}_N^{k} \big),
\end{split}
\end{equation}

\begin{equation}\begin{split}
&(\mathbf{I}_N\otimes\mathbf{1}^\top_M)\big(\mathbf{F}_N\mathbf{D}_N^{k}\mathbf{F}_{N}^{\mathsf{H}}\otimes\tilde{\mathbf{D}}_M^{k}\mathbf{C}_M(l)\big)(\mathbf{e}_{n_p}\otimes\mathbf{e}_{m_p})\approx
\\
&\mathbf{F}_N \big( \big[\mathbf{F}_N^\mathsf{H}\big]_{:,n_p} \odot \mathbf{d}_N^{k} \big).
\end{split}
\end{equation}
Moreover, the term related to data becomes proportional to
\begin{equation}\begin{split}
&(\mathbf{I}_N\otimes\mathbf{1}^\top_M)\big(\mathbf{F}_N\mathbf{D}_N^{k}\mathbf{F}_{N}^{\mathsf{H}}\otimes\tilde{\mathbf{D}}_M^{k}\mathbf{C}_M(l)\big)\mathbf{x}_d\approx
\\
&\sqrt{M}\big(\mathbf{F}_N\mathbf{D}_N^{k}\mathbf{F}_{N}^{\mathsf{H}}\otimes\mathbf{e}_0^\top(\mathbf{D}_M^*)^{l}\mathbf{F}_M\big)\mathbf{x}_d=
\\
&\big(\mathbf{F}_N\mathbf{D}_N^{k}\mathbf{F}_{N}^{\mathsf{H}}\otimes\mathbf{1}_M^\top\big)\mathbf{x}_d,
\end{split}
\end{equation}
where the last equality implies that the data term is proportional to the average of the rows of the data symbol matrix $\mathbf{X}_d$. This reduces the power of the interfering data by a factor of $M$. In fact, as the frame size increases, the data-to-pilot interference tends to zero
\begin{equation}
\lim_{M\rightarrow\infty}\frac{1}{M}\sum_{m=0}^{M-1}[\mathbf{X}_d^\top]_{:,m}=\mathbb{E}\Big[[\mathbf{X}_d^\top]_{:,m}\Big]=\mathbf{0}_N.
\end{equation}
As with the delay profile, the same averaging procedure applies to the noise, thereby reducing its impact.
\end{proof}

\begin{algorithm}[t]
\caption{Proposed fractional delay-Doppler estimation using superimposed cross-pilots}\label{alg:propCEalg}
\KwIn{$\mathbf{R},\ \mathbf{x}_p,\ \{{\theta_p}\}_{p=1}^P$}
\For{$p = 1$ \KwTo $P$}{
$\mathbf{r}_p=\frac{\mathbf{R}\mathbf{a}^*_{rx}(\theta_p)}{N_r}$\;

$\mathbf{y}_p=\big(\mathbf{F}_{N}\otimes\mathbf{I}_{M}\big)\mathbf{r}_p$\;

$\mathbf{v}_p=(\mathbf{I}_N\otimes\mathbf{1}^\top_M)\mathbf{y}_p$\;

$\hat{k}_p = \arg\max_k \big| \mathbf{g}_v^\mathsf{H}(k)\mathbf{v}_p \big|$\;

$\mathbf{u}_p=(\mathbf{1}_N^\top\otimes\mathbf{I}_M)\mathbf{y}_p$\;

$\hat{l}_p = \arg\max_l \big| \mathbf{g}_u^\mathsf{H}(l,\hat{k}_p)\mathbf{u}_p \big|$\;

$\hat{\alpha}_p=\frac{\big(\big(\mathbf{D}_N^{\hat{k}_p}\mathbf{F}_{N}^{\mathsf{H}}\otimes\tilde{\mathbf{D}}_M^{\hat{k}_p}\mathbf{C}_M(\hat{l}_p)\big)\mathbf{x}_p\big)^\mathsf{H}\mathbf{r}_p}{N_pE_p}$\;
}
\textbf{Output:}~$\{\hat{\alpha}_{p},\hat{l}_p,\hat{k}_p\}_{p=1}^{P}$\;
\end{algorithm}

Based on the results presented in Proposition \ref{Proposition:delayprofile} and Proposition \ref{Proposition:dopplerprofile}, disjoint delay-Doppler estimation can be performed as follows.
Since the \gls{CP} prevents \gls{ISI}, as shown in \eqref{eq:DopplerProfileProposition}, the Doppler profile is independent of the delay. Therefore, the Doppler shift can be estimated through a \gls{MF} as
\begin{equation}\label{eq:dopplerxcorr}
\hat{k}_p = \arg\max_k \big| \mathbf{g}_v^\mathsf{H}(k)\mathbf{v}_p \big|.
\end{equation}    
In contrast, due to the presence of \gls{ICI}, the delay profile in \eqref{eq:delayProfileProposition} depends on both Doppler and delay shifts. Thus, the delay can be estimated after obtaining the Doppler estimate as
\begin{equation}\label{eq:delayxcorr}
\hat{l}_p = \arg\max_l \big| \mathbf{g}_u^\mathsf{H}(l,\hat{k}_p)\mathbf{u}_p \big|.
\end{equation}
The maximization in \eqref{eq:dopplerxcorr} and \eqref{eq:delayxcorr} is carried out in two steps: (i) a coarse estimate of the integer part is obtained by identifying the absolute peak of the integrated delay and Doppler profiles; (ii) the search space is narrowed to include fractional delays and Doppler shifts around the initial estimates (specifically $l \in [\hat{l}-0.5, \hat{l}+0.5]$ and $k \in [\hat{k}-0.5, \hat{k}+0.5]$) to maximize the correlation via the \glspl{MF} in \eqref{eq:dopplerxcorr} and \eqref{eq:delayxcorr}.
Once the delay-Doppler pairs are determined, the channel gain of the $p$-th path is obtained via \gls{LS} estimation as
\begin{equation}
\hat{\alpha}_p=\frac{\Big(\big(\mathbf{D}_N^{\hat{k}_p}\mathbf{F}_{N}^{\mathsf{H}}\otimes\tilde{\mathbf{D}}_M^{\hat{k}_p}\mathbf{C}_M(\hat{l}_p)\big)\mathbf{x}_p\Big)^\mathsf{H}\mathbf{r}_p}{N_pE_p}.
\end{equation}
The details of the proposed approach are provided in Algorithm \ref{alg:propCEalg}. 
It can also be noted that the complexity of the proposed estimation algorithm is $O(P(MN N_r + MN \log_2 N))$.
Under the assumption of a large-scale antenna array at the receiver, where the number of antennas $N_r$ is sufficiently high such that $N_r \gg \log_2 N$, spatial processing becomes the dominant factor.Consequently, the complexity simplifies to $O(P MN N_r)$, thus scaling linearly with the system parameters.

\begin{figure*}[t]
    \centering 
    \begin{subfigure}[t]{0.32\textwidth}
\centering
 \resizebox{0.99\columnwidth}{!}{
%
%

\definecolor{mycolor2}{rgb}{0.13, 0.70, 0.29}
\definecolor{mycolor1}{rgb}{0.58, 0.0, 0.83} 
\definecolor{mycolor3}{rgb}{0.92900,0.69400,0.12500}%
\definecolor{mycolor4}{rgb}{0.12941,0.12941,0.12941}%
\begin{tikzpicture}

\begin{axis}[%
width=3in,
height=2.2in,
at={(1.212in,0.618in)},
scale only axis,
xmin=-15,
xmax=5,
xlabel style={font=\color{mycolor4}},
xlabel={Signal-to-Noise Ratio (SNR) [dB]},
ymode=log,
ymin=2*1e-05,
ymax=1,
yminorticks=true,
ylabel style={font=\color{mycolor4}},
ylabel={BER},
axis background/.style={fill=white},
xmajorgrids,
ymajorgrids,
yminorgrids,
legend style={at={(0.01,0.01)}, anchor=south west, legend cell align=left, align=left}
]

\addplot [color=mycolor1, line width=2.0pt, mark=o, mark options={solid, mycolor1}]
  table[row sep=crcr]{%
-15	0.49791015625\\
-12.5	0.480244140625\\
-10	0.349794921875\\
-7.5	0.158115234375\\
-5	0.0730859375\\
-2.5	0.050146484375\\
0	0.0340625\\
2.5	0.027783203125\\
5	0.01845703125\\
};
\addlegendentry{$4$-QAM, baseline \cite{kanazawa2025}}

\addplot [color=mycolor2, line width=2.0pt, mark=o, mark options={solid, mycolor2}]
  table[row sep=crcr]{%
-15	0.22548828125\\
-12.5	0.142568359375\\
-10	0.073388671875\\
-7.5	0.0258203125\\
-5	0.0053515625\\
-2.5	0.000751953125\\
0	0.000146484375\\
2.5	1.953125e-05\\
5	0\\
};
\addlegendentry{$4$-QAM, proposed}

\addplot [dashed, color=mycolor1, line width=2.0pt, mark=o, mark options={solid, mycolor1}]
  table[row sep=crcr]{%
-15	0.501669921875\\
-12.5	0.47734375\\
-10	0.3875146484375\\
-7.5	0.290625\\
-5	0.245576171875\\
-2.5	0.204931640625\\
0	0.1736669921875\\
2.5	0.1519873046875\\
5	0.1391552734375\\
};
\addlegendentry{$16$-QAM, baseline \cite{kanazawa2025}}

\addplot [dashed, color=mycolor2, line width=2.0pt, mark=o, mark options={solid, mycolor2}]
  table[row sep=crcr]{%
-15	0.3366796875\\
-12.5	0.27455078125\\
-10	0.207958984375\\
-7.5	0.1477685546875\\
-5	0.0973193359375\\
-2.5	0.0540478515625\\
0	0.0295654296875\\
2.5	0.01474609375\\
5	0.0086279296875\\
};
\addlegendentry{$16$-QAM, proposed}

\end{axis}

\end{tikzpicture}%
}
 \caption{\gls{BER} vs. \gls{SNR} at fixed \gls{PDR}$=-5$ dB: the proposed scheme achieves a substantially lower BER by accurately accounting for fractional delays and Doppler shifts, while the baseline suffers significant degradation under non-integer channel parameters.} \label{fig:BERvsSNR}
    \end{subfigure}
    \hfill 
    \begin{subfigure}[t]{0.32\textwidth}
 \centering
 \resizebox{0.99\columnwidth}{!}{
%
%
\definecolor{mycolor2}{rgb}{0.13, 0.70, 0.29}
\definecolor{mycolor1}{rgb}{0.58, 0.0, 0.83}
\definecolor{mycolor3}{rgb}{0.12941,0.12941,0.12941}%
\begin{tikzpicture}

\begin{axis}[%
width=3in,
height=2.2in,
at={(1.212in,0.618in)},
scale only axis,
xmin=-20,
xmax=10,
xlabel style={font=\color{mycolor3}},
xlabel={PDR [dB]},
ymode=log,
ymin=4*1e-4,
ymax=1,
yminorticks=true,
ylabel style={font=\color{mycolor3}},
ylabel={BER},
axis background/.style={fill=white},
xmajorgrids,
ymajorgrids,
yminorgrids,
legend style={
    at={(0.99,0.2)}, 
    anchor=south east, 
    legend cell align=left, 
    align=left
}
]
\addplot [color=mycolor1, line width=2.0pt, mark=o, mark options={solid, mycolor1}]
  table[row sep=crcr]{%
-20	0.5009375\\
-17.5	0.49029296875\\
-15	0.386474609375\\
-12.5	0.26833984375\\
-10	0.1300390625\\
-7.5	0.05982421875\\
-5	0.049248046875\\
-2.5	0.038525390625\\
0	0.041796875\\
2.5	0.045595703125\\
5	0.08111328125\\
7.5	0.10439453125\\
10	0.157373046875\\
};
\addlegendentry{$4$-QAM, baseline \cite{kanazawa2025}}

\addplot [color=mycolor2, line width=2.0pt, mark=o, mark options={solid, mycolor2}]
  table[row sep=crcr]{%
-20	0.388583984375\\
-17.5	0.307373046875\\
-15	0.185078125\\
-12.5	0.042705078125\\
-10	0.0028515625\\
-7.5	0.000673828125\\
-5	0.000615234375\\
-2.5	0.000830078125\\
0	0.001494140625\\
2.5	0.0061328125\\
5	0.019345703125\\
7.5	0.04734375\\
10	0.10015625\\
};
\addlegendentry{$4$-QAM, proposed}

\addplot [dashed, color=mycolor1, line width=2.0pt, mark=o, mark options={solid, mycolor1}]
  table[row sep=crcr]{%
-20	0.5001025390625\\
-17.5	0.4940478515625\\
-15	0.4702978515625\\
-12.5	0.361826171875\\
-10	0.2962548828125\\
-7.5	0.231591796875\\
-5	0.1976953125\\
-2.5	0.179150390625\\
0	0.17052734375\\
2.5	0.1920361328125\\
5	0.1976318359375\\
7.5	0.2379052734375\\
10	0.276123046875\\
};
\addlegendentry{$16$-QAM, baseline \cite{kanazawa2025}}

\addplot [dashed, color=mycolor2, line width=2.0pt, mark=o, mark options={solid, mycolor2}]
  table[row sep=crcr]{%
-20	0.44798828125\\
-17.5	0.4136328125\\
-15	0.2444482421875\\
-12.5	0.1258251953125\\
-10	0.082255859375\\
-7.5	0.0510400390625\\
-5	0.048779296875\\
-2.5	0.0544482421875\\
0	0.0692919921875\\
2.5	0.0989892578125\\
5	0.1329443359375\\
7.5	0.17939453125\\
10	0.23017578125\\
};
\addlegendentry{$16$-QAM, proposed}

\end{axis}

\end{tikzpicture}%
}
 \caption{\gls{BER} vs.\ \gls{PDR} at fixed \gls{SNR}$=-2$ dB:
an optimal PDR balances pilot energy (which governs estimation accuracy) against data-symbol energy (which governs detection reliability); the proposed method reaches its minimum \gls{BER} at \gls{PDR}$\approx -5$ dB, compared to \gls{PDR}$\approx -2$ dB for the baseline.} \label{fig:BERvsPDR}
    \end{subfigure}
    \hfill 
    \begin{subfigure}[t]{0.32\textwidth}
        \centering
 \resizebox{0.99\columnwidth}{!}{
%
%
\definecolor{mycolor2}{rgb}{0.13, 0.70, 0.29}
\definecolor{mycolor1}{rgb}{0.58, 0.0, 0.83}
\definecolor{mycolor3}{rgb}{0.12941,0.12941,0.12941}%
\begin{tikzpicture}

\begin{axis}[%
width=3in,
height=2.2in,
at={(1.212in,0.618in)},
scale only axis,
xmin=8,
xmax=18,
xlabel style={font=\color{mycolor3}},
xlabel={PAPR [dB]},
ymode=log,
ymin=1e-4,
ymax=1,
yminorticks=true,
ylabel style={font=\color{mycolor3}},
ylabel={BER},
axis background/.style={fill=white},
xmajorgrids,
ymajorgrids,
yminorgrids,
legend style={at={(0.01,0.01)}, anchor=south west, legend cell align=left, align=left}
]

\addplot [color=mycolor1, line width=2.0pt, mark=o, mark options={solid, mycolor1}]
  table[row sep=crcr]{%
8.54413375461435 0.5009375\\
8.45431571776187 0.49029296875\\
8.60595269489337 0.386474609375\\
8.99658316688811 0.26833984375\\
9.78070969741472 0.1300390625\\
10.7027899150916 0.05982421875\\
11.9279662738512 0.049248046875\\
12.8196949776935 0.038525390625\\
13.82327931461 0.041796875\\
14.3403894061172 0.045595703125\\
14.8227406264722 0.08111328125\\
14.9974463394494 0.10439453125\\
15.0473374329197 0.157373046875\\
};
\addlegendentry{$4$-QAM, baseline \cite{kanazawa2025}}

\addplot [color=mycolor2, line width=2.0pt, mark=o, mark options={solid, mycolor2}]
  table[row sep=crcr]{%
8.45489285170661 0.388583984375\\
8.47662182784545 0.307373046875\\
9.2379683115744 0.185078125\\
10.6218370911546 0.042705078125\\
12.8882306622439 0.0028515625\\
15.0065224177719 0.000673828125\\
17.1145812411337 0.000615234375\\
18.7137137741435 0.000830078125\\
20.1990101095014 0.001494140625\\
21.2148323629007 0.0061328125\\
21.9747493538551 0.019345703125\\
22.4936310458895 0.04734375\\
22.7318927573961 0.10015625\\
};
\addlegendentry{$4$-QAM, proposed}

\addplot [dashed, color=mycolor1, line width=2.0pt, mark=o, mark options={solid, mycolor1}]
  table[row sep=crcr]{%
8.54413375461435 0.5001025390625\\
8.45431571776187 0.4940478515625\\
8.60595269489337 0.4702978515625\\
8.99658316688811 0.361826171875\\
9.78070969741472 0.2962548828125\\
10.7027899150916 0.231591796875\\
11.9279662738512 0.1976953125\\
12.8196949776935 0.179150390625\\
13.82327931461 0.17052734375\\
14.3403894061172 0.1920361328125\\
14.8227406264722 0.1976318359375\\
14.9974463394494 0.2379052734375\\
15.0473374329197 0.276123046875\\
};
\addlegendentry{$16$-QAM, baseline \cite{kanazawa2025}}

\addplot [dashed, color=mycolor2, line width=2.0pt, mark=o, mark options={solid, mycolor2}]
  table[row sep=crcr]{%
8.45489285170661 0.44798828125\\
8.47662182784545 0.4136328125\\
9.2379683115744 0.2444482421875\\
10.6218370911546 0.1258251953125\\
12.8882306622439 0.082255859375\\
15.0065224177719 0.0510400390625\\
17.1145812411337 0.048779296875\\
18.7137137741435 0.0544482421875\\
20.1990101095014 0.0692919921875\\
21.2148323629007 0.0989892578125\\
21.9747493538551 0.1329443359375\\
22.4936310458895 0.17939453125\\
22.7318927573961 0.23017578125\\
};
\addlegendentry{$16$-QAM, proposed}

\end{axis}

\end{tikzpicture}%
}
 \caption{\gls{BER} vs. \gls{PAPR} trade-off, obtained by sweeping the \gls{PDR} over $[-20,\,10]$ dB: at equal \gls{PAPR}, the proposed scheme consistently achieves lower \gls{BER}, demonstrating superior performance despite a higher \gls{PAPR}
for a fixed \gls{PDR} value.} \label{fig:PAPRvsBER}
    \end{subfigure}
    
    \caption{Simulation results for the proposed cross-pilot OTFS scheme and the
baseline superimposed-pilot method of \cite{kanazawa2025}, evaluated over a $P=4$-path vehicular channel at $f_c = 5.9$ GHz with $v_{\max}=500$ km/h, $M=64$
subcarriers, $N=16$ time slots, and $N_r=32$ receive antennas.}\label{fig:SimulationResults}
\end{figure*}

\section{Simulation Results}
To evaluate the performance of the proposed approach, simulations are conducted considering a carrier frequency $f_c = 5.9$ GHz with a subcarrier spacing $\Delta f = 30$ kHz. The \gls{OTFS} system employs $M = 64$ subcarriers and $N = 16$ blocks, while the receiver is equipped with $N_r = 32$ antennas. The wireless channel is characterized by $P = 4$ multipath components with propagation delays $\tau_p = [0, 0.9, 2.4, 3]$ $\mu$s, average path powers $\mathcal{P}_p = [0, -1, -5, -7]$ dB, and \glspl{DoA} $\theta_p = [10^\circ, 42^\circ, -25^\circ, 24^\circ]$. To simulate high-mobility conditions, a maximum speed $v_{\max} = 500$ km/h is assumed, with Doppler shifts $\nu_p = f_c \frac{v_{\max}}{c} \cos(\phi_p)$, where $\phi_p \sim \mathcal{U}[0, 2\pi]$. The proposed approach is compared against the multiple superimposed pilot scheme in \cite{kanazawa2025}; for data detection, a path-wise \gls{MF} is applied, followed by \gls{MRC}.

Fig. \ref{fig:BERvsSNR} illustrates the \gls{BER} performance as a function of the \gls{SNR}. It is observed that the baseline scheme suffers significant performance degradation due to the presence of fractional delays and Doppler shifts. Conversely, the proposed method effectively accounts for fractional channel parameters, thereby achieving a lower \gls{BER}. 
The impact of the \gls{PDR} on communication performance is analyzed in Fig. \ref{fig:BERvsPDR}. A low \gls{PDR} allocates more energy to data symbols, potentially facilitating detection, but reduces the energy available for pilots, which impairs channel estimation accuracy. In contrast, a higher \gls{PDR} enhances pilot energy and estimation precision at the expense of the data \gls{SNR}. This inherent trade-off leads to an optimal \gls{PDR} that minimizes the \gls{BER} \cite{marchese2025robust6gofdmhighmobility}. As shown in Fig. \ref{fig:BERvsPDR}, the minimum \gls{BER} is reached at a \gls{PDR} of approximately $-5$ dB and $-2$ dB for the proposed and baseline methods, respectively.
Finally, Fig. \ref{fig:PAPRvsBER} depicts the trade-off between \gls{PAPR} and \gls{BER} for both schemes across various \gls{PDR} values. Although the proposed scheme exhibits a higher \gls{PAPR} for a fixed \gls{PDR}, it provides a lower \gls{BER} when compared at a target \gls{PAPR} level, demonstrating its superior efficiency.

\section{Conclusion}
In this work, a novel \gls{OTFS} superimposed cross-pilot scheme for \gls{DoA}-aided multi-antenna receivers has been proposed. By exploiting the angular separability of multipaths, it has been demonstrated that the proposed scheme enables the integration of received delay-Doppler matrices; this reduces the effective data-to-pilot interference through averaging, thereby eliminating the need for complex iterative cancellation. Furthermore, a low-complexity disjoint estimation algorithm has been developed to accurately estimate fractional delays and Doppler shifts. Simulation results confirm that the proposed scheme achieves a significantly lower \gls{BER} compared to state-of-the-art superimposed approaches in high-mobility scenarios, while providing a manageable trade-off between \gls{PAPR} and communication performance.


\ifCLASSOPTIONcaptionsoff
  \newpage
\fi

\balance
\bibliographystyle{IEEEtran}
\bibliography{reference}

\end{document}